 \documentclass[final,12pt]{clear2025} 

\usepackage[format=plain]{caption}
\usepackage{graphicx}

\usepackage[utf8]{inputenc} 
\usepackage[T1]{fontenc}    
\usepackage{hyperref}       
\usepackage{url}            
\usepackage{booktabs}       
\usepackage{amsfonts}       
\usepackage{amsmath} 
\usepackage{nicefrac}       
\usepackage{microtype}      
\usepackage{xcolor}         
\usepackage[english]{babel}
\usepackage{times}
\newtheorem{thm}{Theorem}
\newtheorem{defn}{Definition}

\title[Local Interference: Removing Interference Bias in Semi-Parametric Causal Models]{Local Interference: Removing Interference Bias in Semi-Parametric Causal Models}

\clearauthor{%
 \Name{Michael O'Riordan} \Email{moriordan@spotify.com}\\
 \addr Spotify
 \AND
 \Name{Ciar\'{a}n M. Gilligan-Lee} \Email{ciaran.lee@ucl.ac.uk}\\
 \addr Spotify and University College London%
}

\begin{document}

\maketitle

\begin{abstract}
Interference bias is a major impediment to identifying causal effects in real-world settings. For example, vaccination reduces the transmission of a virus in a population such that everyone benefits---even those who are not treated. This is a source of bias that must be accounted for if one wants to learn the true effect of a vaccine on an individual's immune system. Previous approaches addressing interference bias require strong domain knowledge in the form of a graphical interaction network fully describing interference between units. Moreover, they place additional constraints on the form the interference can take, such as restricting to linear outcome models, and assuming that interference experienced by a unit does not depend on the unit’s covariates. Our work addresses these shortcomings. We first provide and justify a novel definition of causal models with \emph{local interference}. We prove that the True Average Causal Effect, a measure of causality where interference has been removed, can be identified in certain semi-parametric models satisfying this definition. These models allow for non-linearity, and also for interference to depend on a unit's covariates. An analytic estimand for the True Average Causal Effect is given in such settings. We further prove that the True Average Causal Effect cannot be identified in arbitrary models with local interference, showing that identification requires semi-parametric assumptions. Finally, we provide an empirical validation of our method on both simulated and real-world datasets.
\end{abstract}

\section{Introduction}

Medical professionals and epidemiologists want to understand the effects of vaccination. Advertisers want to estimate the impact of campaigns in online marketplaces. Teachers want to understand how their instruction influences student's test scores. In all of these cases, straightforward application of causal inference techniques will lead to incorrect estimates of the causal effects due to \emph{interference bias}. Interference bias arises when the treatment assignment of one unit can impact the outcome of another. Indeed, vaccination reduces the transmission of a virus in a population, such that everyone benefits---even those not vaccinated; advertisements compete in online marketplaces, reducing the impact of certain campaigns; and student's test scores are not only influenced by their instruction type, but also by the instruction type of their class-mates.

The restriction that there be no interference between units is a crucial part of the stable unit-treatment value assumption (SUTVA), generally required to identify causal effects \citep{rubin1978bayesian}. Therefore, the presence of interference, and consequent violation of SUTVA, is a major impediment to using causal inference to address real-world problems. 
Recent work by \cite{zhang2022causal,spohn2023graphical,zhang2023causal, o2024spillover} has begun exploring this problem in the graphical models framework of \cite{Pearl2009}. 
However, these approaches require strong domain knowledge---in the form of a graphical interaction network fully describing interference between units---to remove bias. 
Moreover, additional constraints are placed on the form interference can take, such as restricting to linear outcome models \citep{zhang2022causal,spohn2023graphical}, 
and assuming that interference does not depend on a given unit's covariates [\citealt{spohn2023graphical}]. 
To remove interference bias and identify the true causal effect 
(we formally define this quantity in Section~\ref{section: identification of TACE} below), 
these works rely on full knowledge of the interaction network, and also require the existence of units not impacted by interference. These conditions are not always feasible in practice. 

In this paper we address these shortcomings. Our approach is based on two observations.
Firstly, those units who are not themselves treated---that is, units that are only impacted by \emph{spillover effects} from treated units---provide a sense of the level of interference present. If interference experienced by the untreated units is representative of that experienced by the treated units, then in principle one could use spillover effects on untreated units to remove, or reduce, interference bias.

Secondly, even in situations where the interaction network is not fully known, there are real-world settings where the problem can be simplified. Indeed, in the case of vaccination, while it’s possible that the vaccination status of one member of the population can impact the outcome for any other (someone could get vaccinated in London and get on a plane to New York soon after), it is reasonable to assume that a member of the population is primarily impacted by the vaccination status of the members of their immediate household. In the case of advertisements in an online marketplace, a given advertisement will usually only compete with other advertisements that are relevant for the same buyer in that marketplace and not others. Finally, student's test scores will likely only be influenced by their close friends in the class, rather than all class-mates. Hence, even when the interaction network is not fully known, certain domain knowledge about how interference is ``localised'' can make the problem tractable, and allow us to understand if treated units are exposed to the same type of interference as untreated units. 

We show that, by formalising the above two observations, one can remove interference bias in certain semi-parametric settings---beyond linear outcome models---where domain knowledge tells us that the interference can be considered \emph{local} (as in the examples of vaccination, online advertisements, and student test scores above). The main contributions of this paper are as follows:
\begin{enumerate}
    \item A novel definition of structural causal models with \emph{local interference}. In these models interference effects can be confounded, and can depend on unit's covariates. 
    \item A proof that the True Average Causal Effect (TACE), a measure of causality where local interference bias has been removed, can be identified in certain semi-parametric models. An analytic estimand for the TACE is given for such models.
    \item A proof that the TACE cannot be identified in arbitrary models with local interference---hence semi-parametric assumptions are required for identification.
    \item An empirical validation of our method on both synthetic and real-world datasets.
\end{enumerate}


\section{Toy Examples}\label{section: toy examples}

\begin{figure}[t]
\centering
\subfigure[]{
   \includegraphics[width=0.4\textwidth]{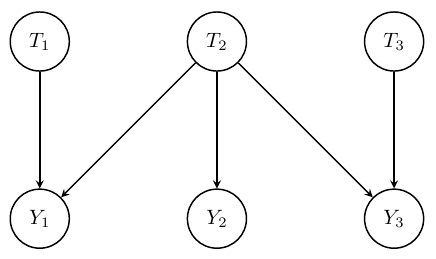}
    }
    \qquad
\subfigure[]{
   \includegraphics[width=0.4\textwidth]{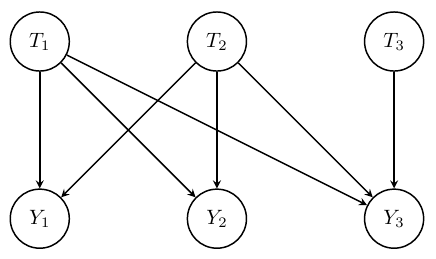}
}
       \caption{(a) Causal structure for toy example 1 described in Section~\ref{section: toy examples}. The spillover effect from unit 2 to unit 3 is used to remove the interference bias and identify the true effect for unit 1. (b) Causal structure for toy example 2 from Section~\ref{section: toy examples}. In this case, the true effect for unit 1 cannot be identified.}
       \label{fig:toy_example_1}
\end{figure}

We now introduce two toy examples which will help build intuition about when one might and might not be able to use spillover effects to remove interference bias. 

Consider the causal model, with causal structure depicted in Figure~\ref{fig:toy_example_1}(a), and structural equations given by $Y_1 = \alpha T_1 + \beta T_2$, $Y_2 = \gamma T_2$, and $Y_3 = \rho T_3 + \beta T_2$. Moreover, assume we have observed that $T_1 = 1$, $T_2 = 1$, and $T_3=0$. That is, units $1$ and $2$ are treated, and unit $3$ is untreated. Our goal is to learn the treatment effect of $T_1$ on $Y_1$ (ie. to identify $\alpha$). Here, units $1$ and $3$ experience interference from unit $2$'s treatment, in the sense that their outcomes $Y_1, Y_3$ depend on $T_2$. There are spillover effects from unit $2$'s treatment to the outcomes of units $1$ and $3$. However, as unit $3$ is untreated, observing $Y_3$ yields the spillover effect directly and gives us $\beta$. Hence, subtracting this from $Y_1$ provides us with $\alpha$: $Y_1 - Y_3 = \alpha + \beta - \beta= \alpha$. We have thus removed the interference bias to learn the true effect of treatment on unit $1$ alone using the spillover effect on unit $3$. This required that both unit $1$ and unit $3$ only experienced interference from unit $2$, and moreover that they both experienced the same \emph{level} of this interference---that is both depended on $T_2$ through $\beta$. 

The expression used to obtain $\alpha$ is quite interesting, and worth considering in more detail. We had domain knowledge informing us that unit $1$ and unit $3$ experienced the same type and degree of interference, and thus could consider the untreated unit $3$ a ``match'' for the treated unit $1$---in a similar manner to matching on confounders in propensity-based causal inference. In fact, this observation foreshadows our approach to removing interference bias provided in Section~\ref{section: identification of TACE}.

In our next toy example we will see that these assumptions are crucial for removing interference bias. Consider the causal model, with causal structure depicted in Figure~\ref{fig:toy_example_1}(b), and structural equations given by $Y_1 = \alpha T_1 + \beta T_2$, $Y_2 = \gamma T_2 + \delta T_1$, and $Y_3 = \rho T_3 + \beta T_2 + \delta T_1$. As before, assume we have observed that $T_1 = 1$, $T_2 = 1$, and $T_3=0$. Again, our goal is to identify $\alpha$. Here however, it is not possible to use spillover effects to do this. Even though unit $1$ and unit $3$ are impacted by interference from unit $2$ to the same degree (both through $T_2$ via the coefficient $\beta$), unit $3$ is also impacted by interference from unit $1$, which means the type of interference experienced by unit $3$ is no longer representative of the interference experienced by unit $1$ (or unit $2$ for that matter).

In Section~\ref{section: local interferencee} we provide and justify a novel definition of causal models with \emph{local interference}, and prove that interference bias can be removed in certain semi-parametric models satisfying this definition using spillover effects from untreated units. First, we formally overview the structural causal model and interaction network framework that we work in for the rest of the paper. 

\section{Preliminaries: Interaction Networks and the True Average Causal Effect}
\label{section: prelim}

We adopt the Structural Causal Model (SCM) framework introduced by \cite{Pearl2009}. 

\begin{defn}[Structural Causal Model] \label{functional causal model}
\label{scmdef}
A structural causal model (SCM) specifies a set of latent variables $U=\{u_1,\dots,u_n\}$ distributed as $P(U)$, a set of observed variables $X=\{X_1,\dots, X_n\}$, a directed acyclic graph (DAG) $G$, called the \emph{causal structure} of the model, whose nodes are the variables $U\cup X$, a collection of functions $F=\{f_1,\dots, f_n\}$, such that $X_i = f_i(\text{PA}(X_i), u_i), \text{ for } i=1,\dots, n,$ where $\text{PA}$ denotes the parent observed nodes of an observed variable.
\end{defn}

The latent noise term in each $f_i$ can be suppressed into $\text{PA}(X_i)$ by enforcing that every observed node has an independent latent variable as a parent in $G$. This convention is adopted in this work. 

A (hard) intervention on variable $X_{i}$ is denoted by $\text{do}(X_{i}=x_{i})$, which corresponds to removing all incoming edges in the causal graph and replacing its structural equation with a constant.  

We work in the interacting models framework of \cite{zhang2022causal} (see also \cite{zhang2023thesis}). We refer to the variables in a standard causal model as \emph{generic variables}. An \emph{explicit variable}, on the other hand, represents a variable corresponding to a specific unit only. For instance, ``treatment ($T$)'' is a generic variable, but ``the treatment assignment of unit $i$ ($T_i$)'' is explicit.

\begin{defn}[Interaction model]
An \emph{interaction model}, $M(G, S)$, is a causal model where $G$ is the \emph{interaction network} and $S$ is the set of structural equations defining the data generating process of the observed explicit variables. An interaction network, $G$, is a DAG with each node representing an explicit variable and each directed edge $A_i \rightarrow B_j$  representing $A_i$ causes $B_j$
\end{defn}

An example of an interaction network is given by Figure~\ref{fig:toy_example_1}(a), with structural equations specified in Section~\ref{section: toy examples}. Interaction networks allow arrows between explicit variables of the same unit ($T_1\rightarrow Y_1$ in Figure~\ref{fig:toy_example_1}(a)), as well as between explicit variables of different units ($T_2\rightarrow Y_1$ in Figure~\ref{fig:toy_example_1}(a)).

We now define the \emph{default interaction model}, which is the ``default'' causal model for a unit if there are no interactions with other units. Our \emph{default} interaction model extends the \emph{isolated} interaction model of \cite{zhang2022causal}, allowing for non-linear interactions between units in the structural equations.
Interacting units are removed by replacing the variables in the interacting terms in the structural equations with the appropriate \emph{default value}. In linear models for instance, removing interference corresponds to removing the terms that lead to interactions between units, and so default values are always $0$ in this case. Concretely, recall the structural equations for unit $1$ in the second toy example from Section~\ref{section: toy examples}: $Y_1=\alpha T_1 + \beta T_2$. The default model for unit $1$ then corresponds to setting $T_2=0$. The structural equation in the default model for unit $1$ in this toy example is: $Y_1=\alpha T_1$. In non-linear models, default values can be any real number that the variable can take on in order to remove interference. For example, if the structural equations are a product of variables from different units, then to remove interference the default value for interacting variables must be $1$.

\begin{defn}[Default interaction model]
A default interaction model $DM(DG, DS)$ with respect to an interaction model $M(G, S)$ with a set of default generic variable values $\text{D}$ is constructed from $M$ in the following way:
\begin{enumerate}
\item $DG = G'$ where $G'$ is the graph obtained from $G$ by removing all edges between units $i\neq j$.
\item $DS = S'$ where $S'$ is the set of equations obtained from $S$ by substituting in each equation $X_i=f_i(\text{PA}(X_i))$ any explicit variable $\forall j \neq i$ with the corresponding constant in $\text{D}$. 
\end{enumerate}
The unit default interaction model $DM_i(DG_i,DS_i)$ corresponds to the subgraph and subset of equations in $DM(DG,DS)$ specific to unit $i$.
\end{defn}

In order to define causal quantities that are the same for all units, we need to introduce a notion of symmetry in interaction models. To this end, we define \emph{balanced interaction models}. The intuition underpinning this definition is that if we hypothetically remove all interactions, then all units should have the same data generating process---as in standard causal inference.

\begin{defn}[Balanced interaction model]\label{def: balanced}
An interaction model $M$ is balanced with default values $\text{D}$ if the unit default interaction model $DM_i(DG_i,DS_i)$ is identical for each unit.
\end{defn}

Finally, we define the causal quantity of interest in this work: the True Average Causal Effect (TACE). \citet{zhang2022causal} introduced the TACE as a generalization of the usual average causal effect to the non-iid setting. The TACE is similar to the \emph{direct effect} defined by \citet{hudgens2008toward} in the potential outcomes framework. We work in the interacting models framework and so use the TACE terminology for consistency with \citet{zhang2022causal}. This terminology also avoids potential confusion with the notion of direct vs indirect effects in causal graphs with only generic variables, where the path $T\rightarrow Y$ is often referred to as the direct effect, and $T\rightarrow M\rightarrow Y$ as an indirect effect mediated by $M$.

\begin{defn}[True Average Causal Effect]
Let $M$ be a balanced interaction model with default values $\text{D}$. The True Average Causal Effect of binary treatment $T$ on outcome $Y$, denoted as \emph{TACE}$_{TY}$, is the Average Causal Effect of $T$ on $Y$ in the unit default model of $M$ with default values $\text{D}$. That is \emph{TACE}$_{TY}$ corresponds to $\mathbb{E}\left(Y \mid \text{do}(T=1)\right) - \mathbb{E}\left(Y \mid \text{do}(T=0)\right)$ in $DM(DG, DS)$.  
\end{defn}

\section{Local Interference} \label{section: local interferencee}

Consider the interaction model depicted in Figure~\ref{fig:local_int}. The outcome for unit $i$, $Y_i$, depends on the treatment status, $T_i$, and covariates, $X_i$, of unit $i$, but also the treatment status $T_j$ and covariates $X_j$ of the $N-1$ other units where $j\neq i$. For the rest of this paper, we assume that all relevant covariates are observed, and that there are no unobserved confounding factors. Denote the set of all $N-1$ treatments $T_j$ for $j \neq i:=\{T_1,\dots,T_{i-1},T_{i+1},\dots,T_N\}$ by $T_{j\neq i}$, and the set of all $N-1$ covariates $X_j$ for $j \neq i:=\{X_1,\dots,X_{i-1},X_{i+1},\dots,X_N\}$ by $X_{j\neq i}$. We can then write $Y_i=f_i(T_i,X_i, T_{j\neq i}, X_{j\neq i})$, where we have suppressed the latent noise terms.

\begin{figure}[t] 
\centering
   \includegraphics[scale=0.9]{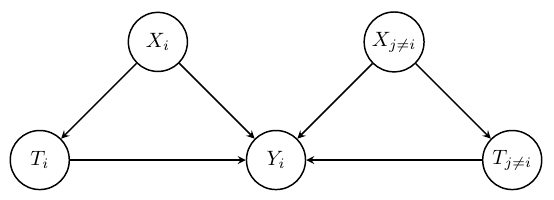}
       \caption{General causal structure with interference between units. Nodes and arrows from different units are compressed for ease of representation. The latent noise terms are also not shown.
       }\label{fig:local_int}
\end{figure}

Without loss of generality, we assume there exist functions $h_i$ and $q_i$ such that we can write $Y_i=h_i(T_i,X_i,q_i(T_{j\neq i},X_{j\neq i}, T_i, X_i))$. There is no loss of generality here as $T_i$ and $X_i$ are inputs to both $h_i$ and $q_i$. We can think of the function $q_i$ as governing the degree of interference experienced by unit $i$. With this re-framing, we can modify the causal structure to the DAG depicted in Figure~\ref{fig:local_int_2}, such that the interference experienced by unit $i$ is mediated by a variable $I_i$, where $Y_i=h_i(T_i,X_i,I_i)$. 

\begin{figure}[t] 
\centering
   \includegraphics[scale=0.9]{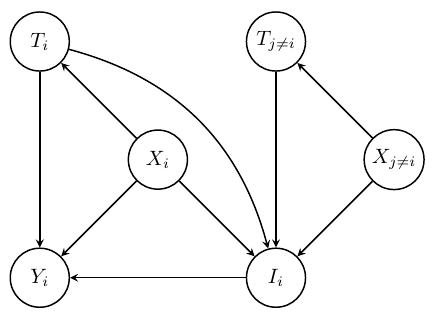}
      \caption{Graphical reduction of interference to unit-specific information. 
       }\label{fig:local_int_2}
\end{figure}

In order for our interaction model to be balanced in the sense of Definition~\ref{def: balanced}, we require $Y_i=h(T_i,X_i, I_i)$, where $h$ has no index. That is, we need the set of variables $\{Y_i,T_i,X_i,I_i\}$ and equations between them to be the same for each unit. Note that, while the set of variables $\{Y_i,T_i,X_i,I_i\}$ is the same for all units, the values those variables take need not be the same.

To have a chance at using spillover effects on the untreated units to remove interference bias on the treated units and isolate the causal effect of $T_i$ on $Y_i$, we must impose an overlap assumption: $0<P(T_i|X_i,I_i)<1$. This is an extension of the standard overlap assumption used in propensity-based causal inference to ensure that not only are there similar units in terms of covariate values between the treated and untreated groups, they are also similar in terms of the interference they experience. Such an overlap condition holds in the first toy example from Section~\ref{section: toy examples}. But, as demonstrated in the second toy example, this can fail to hold even in the simple linear case, with no confounders and constant interference effects where the effect of $T_i$ is the same across units $j\neq i$. Concretely, in the first toy example units 1 and 3 both experience interference level $I=\beta$, with treated unit 1 and untreated unit 3, and so we have that $P(T=1\mid I=\beta)=0.5$.
Whereas in the second toy example only treated unit 1 experiences interference level $I=\beta$, and there is no corresponding match in the untreated group. This violates the overlap inequality with $P(T=1 \mid I=\beta)=1$.

Overlap will not hold in all models and settings. To make the problem more tractable, we introduce the notion of \emph{local interference}, where the basic idea is that in some real-world settings the interference experienced by unit $i$ is likely determined by a subset of the $N-1$ other units, and that many units should share the same level and type of interference. This is the case in the first toy example, where units $2$ and $3$ don’t experience interference from unit $1$, and the interference experienced by those units is the same. In situations where this notion is valid, we refer to the mediator $I_i$ as the \emph{local interference signature}. 
We can think of the interference signature as a key that selects the correct ``type'' of interference for each unit. 
Note that two units can have the same interference signature $I_i$, but still experience different levels of interference depending on their covariate values.

In practice, the assumption of \emph{local interference} means that domain knowledge should give us a simplified or approximate model of interference such that the resulting interference signature has reduced cardinality relative to all \emph{a-priori} possible $I_i$. For instance, by constraining the relevant interfering units to be those a certain distance away, or, as in the vaccine example from the introduction, to be from the same household, and so on. An extreme example would be where the domain knowledge dictates that the interference experienced by a unit is determined by the total number of people in their household, and so in this case the interference for unit $i$ would be determined by unit $i$’s covariates alone without having to look at other units. In this extreme example of local interference, the overlap requirement reduces to the usual assumption $0<P(T_i|X_i)<1$.

Using domain knowledge and feature engineering to specify $I_i$ for a given problem setting, then, is much like using domain knowledge and feature engineering to specify relevant confounders. We discuss some examples and their corresponding $I_i$'s in Section~\ref{section: specifying the local interference signature}. In some cases, it is straightforward to specify $I_i$. In others, it may not be possible---just as there are real-world situations where confounders are hidden and can't be adjusted for. 

We now define models with \emph{local interference} where this is possible and where overlap holds.

\begin{defn}[Interaction model with local interference]\label{Def: local}
An \textit{interaction model with local interference} is a balanced interaction model with interaction network given in Figure~\ref{fig:local_int_2}, structural equations for each unit $i$ given by 
\begin{equation}
Y_i=h(T_i,X_i,I_i,\sigma_{Y_i}), \quad T_i=l(X_i, \sigma_{T_i}),\quad X_i=m(\sigma_{X_i}), \quad I_i=q(T_i, X_i, T_{j\neq i},X_{j\neq i},\sigma_{I_i}), 
\end{equation}
where the latent noise terms $\sigma$ are all drawn i.i.d., 
and the following overlap condition is satisfied: 
\begin{equation}
0<P(T|X,I)<1.
\end{equation}
\end{defn}

As discussed at the start of this Section, the structural equations in the above Definition hold for any (balanced) model with the interaction network as in Figure~\ref{fig:local_int} without loss of generality. The main ingredient of Definition~\ref{Def: local} is the overlap assumption. This is what is meant by \emph{local} interference. 

In Section~\ref{section: identification of TACE}, we show that that certain semi-parametric assumptions allow the TACE to be computed in interaction models with local interference by appropriately comparing treated and untreated units that experience the same levels of interference and confounding. First, we provide some examples showing how to specify local interference signatures using domain knowledge in real-world settings.

\subsection{Specifying the local interference signature} \label{section: specifying the local interference signature}

In this Section, we discuss specifying local interference signatures in two of the examples from the introduction: vaccinations and student test scores.

Consider the case of vaccinations. It is reasonable to assume that an individual is mainly impacted by the vaccination status of their immediate household. In this case the local interference signature could correspond to the fraction of an individual's household that are vaccinated. This is a function of other unit's treatment assignments---the treatment status of the rest of the household---as well as the value of one of that individual's covariates---the number of members in their household.  


In the case of student test scores, it is reasonable to assume that students will most likely be influenced by their close friends in the class, rather than all class-mates. Here, the local interference signature could correspond to the list of a student's close friends, the teachers those friends have, as well as the covariates of those friends relating to their prior educational aptitude and previous test scores. This is a function of both the student's treatment assignment as well as that of their close friends, and also certain covariate values of the student and their friends.

We discuss local interference signatures further for the experiments in Section~\ref{section: experiments}.

\subsection{Identifying the TACE in models with additive local interference}\label{section: identification of TACE}

Consider the following semi-parametric form for an interaction model with local interference:
\begin{equation}
Y_i = f(T_i, X_i) + g(I_i, X_i) + \epsilon_i, \text{ with } \epsilon_i \sim \mathcal{N}(0, \sigma).
\end{equation}
The treatment and covariates, and the interference signature and covariates can separately interact in an arbitrary, non-linear fashion to give rise to the outcome. Importantly, in the outcome model, the treatment and interference signature only functionally interact in an additive manner\footnote{The interference and treatment can still interact in a non-linear fashion as long as this behaviour is captured in the interference signature $I_i$. Once $I_i$ is specified, further interactions with $T_i$ in the outcome model must be additive.}. 

In such models, the $\text{TACE}_{TY}$ corresponds to: $ \mathbb{E}\left( f(T=1,X) - f(T=0, X)\right)$. We now show that $\text{TACE}_{TY}$ is identifiable from observations of $\{Y_i, T_i, X_i, I_i\}$.

\begin{thm}\label{theorem: IPW}
In an interaction model with local interference and structural equations
$$Y_i = f(T_i, X_i) + g(I_i, X_i) + \epsilon_i, \text{ with } \epsilon_i \sim \mathcal{N}(0, \sigma)$$    
the True Average Causal Effect is identifiable from observations $\{Y_i, T_i, X_i, I_i\}$ and is given by:
\begin{equation}
\emph{\text{TACE}}_{TY} = \mathbb{E}\left(\frac{\mathbb{I}_{T=1}Y}{P(T \mid X, I)}  - \frac{\mathbb{I}_{T=0}Y}{1 - P(T \mid X, I)} \right)
\end{equation}
where $\mathbb{I}_{T=t}$ is an indicator random variable that takes on the value $1$
if $T=t$ and $0$ otherwise.
\end{thm}

\begin{proof}
\begin{align*}
&\mathbb{E}\left(\frac{\mathbb{I}_{T=1}Y}{P(T \mid X, I)}  - \frac{\mathbb{I}_{T=0}Y}{1 - P(T \mid X, I)} \right) \\
&= \mathbb{E}\big(\mathbb{E}\left(Y \mid T=1, X, I \right) - \mathbb{E}\left(Y \mid T=0, X, I \right)\big) \\
&= \mathbb{E}\big( f(T=1, X) + g(I,X) - f(T=0, X) - g(I,X)\big)\\
&=\mathbb{E}\big( f(T=1, X) - f(T=0, X)\big) = \text{TACE}_{TY}.
\end{align*}
The first equality follows from Bayes' rule and the definition of expectations. See Appendix~\ref{section: proof1} for this derivation.
\end{proof}

The expression for $\text{TACE}_{TY}$ in Theorem~\ref{theorem: IPW} is well-defined due to the overlap assumption from Definition~\ref{Def: local}. This expression shows that once the interference signature has been specified for a given problem, interference bias can be removed in estimates of the causal effect by treating the interference signature in a similar manner to a confounder, and adjusting for it accordingly. 

\subsection{Non-identifiability of the TACE in models with non-additive local interference}

Can we identify $\text{TACE}_{TY}$ in more general models than the additive ones explored in the previous Section? We now show that the answer is in general no---showing that semi-parametric assumptions are required for identification. Recall that a quantity is identifiable from a specific type of data if every model that agrees on that data produces the same value for the quantity. Hence, if two models agree on the data, but not on the quantity, then it is not identifiable from that data.

\begin{thm}
The True Average Causal Effect is \textbf{not} identifiable \textbf{in general} from observations of $\{Y_i, T_i, X_i, I_i\}$ in interaction models with local interference when the structural equation for $Y_i$ is not additive: $$Y_i \neq f(T_i, X_i) + g(I_i, X_i) + \epsilon_i, \text{ with } \epsilon_i \sim \mathcal{N}(0,\sigma).$$
Note that this does not imply that the True Average Causal Effect is \textbf{never} identifiable in interaction models with non-additive local interference.
\end{thm}
\begin{proof}
By the definition of identifiability, all we need to do to prove this Theorem is to find two non-additive models that agree on the observed data, but not on $\text{TACE}_{TY}$. Consider these models:
$$\text{Model 1: }Y_i = \alpha T_i\,I_i + X_i, \qquad \text{Model 2: } Y_i = {\alpha}' T_i\,I_i + \beta T_i + X_i,$$
where both models have the same data generating mechanisms for $T_i$, $X_i$, and $I_i$. If we set $\alpha I_i = {\alpha}'I_i + \beta$, then both models have the same joint observational distribution over $\{Y_i, T_i, X_i, I_i\}$, and thus agree on this data. The $\text{TACE}_{TY}$ in the first model is $\alpha$ and the $\text{TACE}_{TY}$ in the second model is $\alpha' + \beta$. As long as $I_i$ does not equal $1$ for at least one unit, then these models have different $\text{TACE}_{TY}$. That is, they agree on the data, but not on the causal estimand, hence one cannot identify the $\text{TACE}_{TY}$ from observational data alone.
\end{proof}

This is not to say that non-additive functional interactions between $I_i$ and $T_i$ in the outcome model always preclude the identification of causal quantities. 
For example, in product outcome models with local interference $$Y_i = f(T_i, X_i)g(I_i, X_i) + \epsilon_i, \text{ with } \epsilon_i \sim \mathcal{N}(0, \sigma),$$ the True Average Causal Risk Ratio corresponds to  $\text{TACRR}_{TY}=\mathbb{E}\left(f(T=1, X)/f(T=0, X)\right)$. In a similar manner to the additive case, we can estimate TACRR$_{TY}$ by conditioning on $X_i$ and $I_i$, matching treated and untreated units in each segment.
\begin{equation*}
\text{TACRR}_{TY} = \mathbb{E}\left(\frac{f(T=1,X)}{f(T=0,X)}\right)
= \mathbb{E}\left(\frac{f(T=1,X)g(I,X)}{f(T=0,X)g(I,X)}\right)
= \mathbb{E}\left(\frac{\mathbb{E}(Y\mid T=1,X,I)}{\mathbb{E}(Y\mid T=0,X,I)}\right)
\end{equation*}

%
%

\section{Related work}

Interference was first formally defined by Cox in his 1958 work \citep{cox1958planning}.  Handling interference is non-trivial, and most works in causal inference assume no interference by invoking SUTVA \citep{rubin1978bayesian}. 
Much work has been done to relax this assumption, facilitating the investigation treatment effects in the presence of interference using the concept of exposure mappings \citep{hudgens2008toward,manski2013identification,aronow2016estimating,eckles2017design,chin2019regression,auerbach2021local,savje2021average,leung2022causal,savje2024causal}.
The main idea is to assume the existence of functions mapping the treatment and network information into an \emph{effective treatment} for each unit, such that the potential outcomes depend only on the treatment assignments via this compressed representation.

Recent years have also seen a rise in the use of graphical tools to explore and remove the bias introduced by interference \citep{ogburn2014causal, shpitser2017modeling,sherman2018identification}. These works rely on a notion of \textit{partial interference} which divides units into equal-sized blocks and assumes that interference can only occur within a block but not across different blocks (see also \citealt{liu2014large}). 
More recent works by \cite{zhang2022causal, zhang2023causal} (see also \citealt{zhang2023thesis}),  \cite{yu2022estimating}, \cite{spohn2023graphical}, and \cite{ogburn2024causal} relax these assumptions. 
Our work closely follows \cite{zhang2022causal} and \cite{spohn2023graphical}, aiming to overcome specific limitations of these works, and explores how to remove interference bias when interference effects are confounded, can depend on a unit's covariates, and when structural equations aren't linear.

The approach of \cite{zhang2022causal} requires strong graphical knowledge to remove bias. In particular, they require the existence of units in the interaction network whose outcomes are not impacted by other unit's treatment assignment. In a follow-on work, \cite{zhang2023causal} relax the strict graphical requirements, and replace them with uncertainty estimates of the graphical structure. While this reduces the type of domain knowledge required, assigning a quantitative level of uncertainty over graphical interaction networks nevertheless requires a reasonable degree domain knowledge. Additionally, they only consider linear structural causal models. 

The approach of \cite{spohn2023graphical} also restricts to linear structural causal models. 
Although this strict linear assumption does in principle allow for identification of the TACE, the approach is actually designed around identification of a different causal quantity, namely the Global Average Treatment Effect (GATE).
Moreover, the model does not allow for interference effects to depend on a unit's covariates, nor does it allow the interference effects from other units to be confounded.

\section{Experiments} \label{section: experiments}
\subsection{Simulated data} \label{subsec: sim data}

In this Section, we demonstrate our method on data sampled from an additive outcome model with local interference. 
We construct 10,000 simulated datasets according to the data generating process outlined below. Each dataset consists of 110,000 units, with a True Average Causal Effect of $\text{TACE}_{TY}=1$. 
The variable $X$ determines each unit's baseline outcome $Y$, and baseline probability of treatment $T$.
Units are randomly\footnote{The notation $V_i\in_R S$ indicates that variable $V$ is sampled uniformly at random from the set $S$.}
assigned to one of 10,000 contexts in $L$.
These contexts are analogous to a unit's local neighbourhood, impacting both the unit's treatment status, and also the level of interference they experience. 
For example, this local context could correspond to a person's household or town, or in the case of online advertising the local context could label distinct product categories such that items in the same category could interference with each other. 
$T$ is sampled as a Bernoulli trial with probability set by $X$ and the local context $L$. Note that units interfere within their own local context such that units with higher proportions of other treated units the same context are more likely to be impacted by interference.
\begin{gather*}
X_i\in_R\{0.1, 0.2, 0.3, 0.4\},\qquad
L_i\in_R\{1,\ldots,10000\},\qquad
U_{L_i}\in_R\{0.1, 0.2, 0.3, 0.4\},\\
T_i\sim \mathcal{B}\left(X_i+U_{L_i}\right),\quad
I_i=\frac{\sum_{j\neq i, L_j=L_i}T_j}{\sum_{j\neq i,L_j=L_i} 1},\quad
W_i\sim\mathcal{B}\left(I_i\right),\quad
Y_i\sim \mathcal{N}\left(4X_i + T_i + 4W_i, 1\right)
\end{gather*}

\begin{figure}[t] 
\centering
   \includegraphics[scale=0.7]{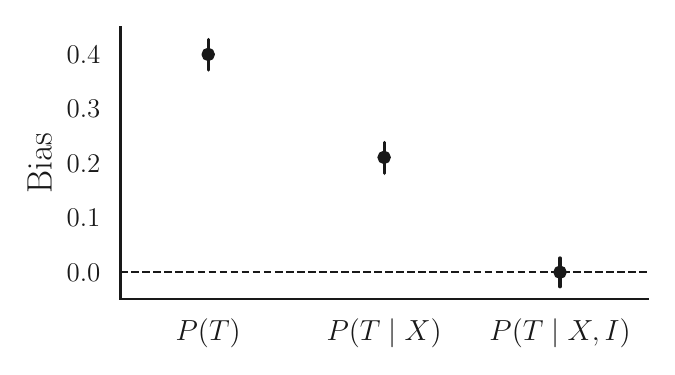}
      \caption{Bias in TACE$_{TY}$ estimates for the simulated data described in Section~\ref{subsec: sim data}. We show the bias when TACE$_{TY}$ is estimated without adjustment (IPW with $P(T)$), adjusting only for $X$ (IPW with $P(T\mid X)$), and adjusting for both $X$ and $I$ (IPW with $P(T\mid X, I)$).
      The error bars show the intervals between the 2.5\% and 97.5\% percentiles across 10,000 simulated datasets.
       }\label{fig:toy_model_bias}
\end{figure}

In Figure~\ref{fig:toy_model_bias}, we show the bias when $\text{TACE}_{TY}$ is estimated without adjustment, adjusting only for the confounders $X$, and adjusting for both $X$ and $I$. In Appendix~\ref{section: more sims}, we show a modified version of this example where the interference signature also depends on $T_i$.

\subsection{Semi-synthetic data} \label{subsec: real data}
\begin{figure}[t] 
\centering
   \includegraphics[scale=0.58]{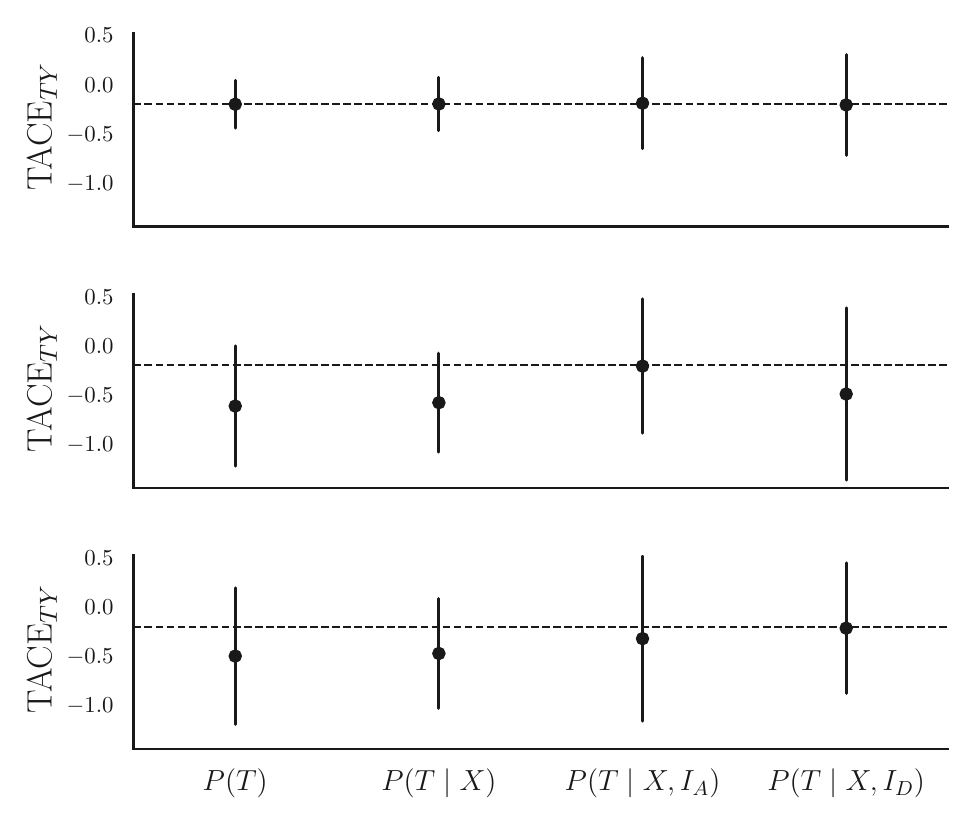}
      \caption{Estimated $\text{TACE}_{TY}$ of restrictive policy adoption in the real-world data from Section~\ref{subsec: real data}.
      The top panel shows the estimates on the unmodified data. The horizontal dashed line is centered on the $P(T\mid X)$ estimate for comparison with other panels. The middle panel shows estimates on a semi-synthetic variant where we introduce inference between \emph{adjacent} cantons, while the bottom panel shows estimates where we introduce interference inversely proportional to the squared \emph{distance} between canton capitals.
      Conditioning on the correctly specified interference signature recovers the effect estimate from the unmodified data, while conditioning on the misspecified signature (eg. $I_A$ instead of $I_D$ or vice versa) still reduces interference bias, but does not completely remove it.
      The error bars show the 95\% bootstrap confidence intervals calculated by resampling cantons, with the variance inflated as described in Appendix~\ref{section: real data appendix}. 
       }\label{fig:covid_tace}
\end{figure}
In this Section, we demonstrate our method on real-world data. We apply our model to estimate the TACE of restrictive policy adoption on the spread of COVID-19 in Switzerland between July 2020 and December 2020. 
Early in the pandemic, the cantons of Switzerland could choose to adopt more restrictive policies than the government-mandated baselines. As people commute between cantons, the policy adoption in one canton could impact the spread of COVID-19 in others---a potential source of interference bias. 
The data and causal assumptions are adapted from \cite{nussli2024effect}\footnote{\url{https://github.com/enussl/Facial-Mask-Policy-COVID-19}}, and \cite{spohn2023graphical}\footnote{\url{https://github.com/henckell/InterferenceCode}}. 
These works estimate various causal quantities corresponding to the effects of introducing a strict facial-mask policy (ie.\ more restrictive than the government-mandated baseline) on the spread of COVID-19 in Switzerland between July 2020 and December 2020. 

To isolate the effect of the facial-mask policy, the authors control for additional policies including closing of workplaces, restrictions on gatherings, and cancellations of public events.
However, adoption of the facial-mask policy is highly correlated with these other policies, resulting in severe overlap violations. Therefore, isolating the effect of the facial-mask policy is only possible with strict parametric assumptions. 
Instead, we estimate the combined effect of these policies by aggregating them into a single binary treatment, where $T_i=1$ in canton $i$ if the canton adopted at least one of the following (a) strict facial-mask policy, (b) required work closures for some sectors, (c) restrictions on gatherings of more than 100 people, or (d) required cancellations of public events. 
We estimate the TACE of restrictive policy adoption on the weekly case growth rate.

Following \cite{spohn2023graphical}, we consider the following confounders: canton population, \% people of age 80+ in a canton, people per km$^2$ in a canton, public school holidays,
and information about the pandemic available to the public in week $t$, given by the lagged
case growth rate at week $t-2$. 
Furthermore, we consider two potential interference signatures: $I_A$, and $I_D$.
$I_A$ is defined in a similar manner to the ``interference feature'' of \cite{spohn2023graphical}. For each canton, $I_A$ is calculated as the average treatment status of \emph{adjacent} cantons. 
$I_D$ is based on the \emph{distance} between cantons, and for each canton is calculated as a weighted average of the treatment status of other cantons, weighted by the inverse squared distance between canton capital cities. 
This second example shows how the interference signature allows for a richer range of interference behaviours than can be captured by the interaction network structure and treatment status alone. 

In Figure~\ref{fig:covid_tace}, we show the estimated $\text{TACE}_{TY}$ of restrictive policy adoption on the weekly case growth rate. 
We adjust for an interference signature $I_A$ defined based on the treatment status of \emph{adjacent} cantons, and also for $I_D$ based on treatment status and \emph{distance} between canton capitals. 
On the unmodified data (top panel), we find no evidence of significant interference bias, and so we also construct two semi-synthetic variants where we artificially introduce interference between cantons, which we describe below.

The interference signature $I_A$ is defined as the fraction of adjacent cantons that are treated. For a given canton, we further multiply this by the population density of that canton, such that cantons with larger population density, and with larger fraction of adjacent treated cantons, experience higher levels of interference. In the outcome model, the introduced synthetic interference is therefore a non-linear (multiplicative) combination of the interference signature and canton covariates, which is added to the original outcome --- recall that this additive interference is a requirement for identifiability in our semi-parametric model.
We introduce synthetic interference for the $I_D$ case in a similar manner, such that cantons with larger population density, larger fraction of adjacent treated cantons, and with capital cities close to other treated capital cities, experience higher levels of interference. This demonstrates how, in our framework, covariates of other units --- in this case distance between units --- can play a role in the interference experienced by a given unit.

How is it that the outcome model of \cite{spohn2023graphical} (see their Eq. 4) allows for an interaction term that is a product between the treatment and interference, whereas our outcome model must restrict to additivity between the treatment and interference?
Firstly, it should be noted that \cite{spohn2023graphical} aim to identify the Global Average Treatment Effect (GATE), which corresponds to the sum of the TACE and interference effects, and so for their purposes it's not strictly necessary to identify the TACE. 
Secondly, the TACE can be identified in linear interference models with product interaction terms in the outcome model if either (a) there exist treated and untreated units who experience zero interference (as assumed by \citealt{zhang2022causal}), or (b) the linear model correctly extrapolates beyond the support of the data to this zero-interference setting (as assumed by \citealt{spohn2023graphical}). 
Finally, as discussed in Section~\ref{section: identification of TACE}, our model does in fact allow for non-additive interactions between the treatment and interference, provided these are captured by the interference signature such that overlap holds.

%
%

\section{Conclusion}
The ability to answer causal questions is vital for addressing real-world decision-making problems \cite{richens2020improving, fawkes2025hardness, vlontzos2023estimating, zeitler2023non, jeunen2022disentangling, van2023estimating, andreu2024contrastive}. However, the existence of interference bias is a barrier to answering causal questions. We provided and justified a novel definition of causal models with \emph{local interference}. We proved that the TACE can be identified in certain semi-parametric models satisfying this definition, where interference effects are additive in the outcome model, but the interaction of interference and a unit's covariates could be non-linear. An analytic estimand for the TACE was given in such settings. We further proved that the TACE cannot be identified in arbitrary models with local interference, showing that identification requires semi-parametric assumptions. However, we demonstrated that other causal quantities, such as the TACRR, can be identified in certain models where the semi-parametric additivity requirement on the outcome model fails. Finally, we provided an empirical validation of our method on both simulated and real-world datasets.

\bibliography{main_clear_2025}

\section*{Appendix}
\appendix

\section{Proof of Theorem~\ref{theorem: IPW}}
\label{section: proof1}
In this section, we write out in full the steps involved in the first line of the proof of Theorem~\ref{theorem: IPW}.
$$
\begin{aligned}
\mathbb{E}( \mathbb{E}(Y| T=1, X, I)) &= \mathbb{E}\left( \sum_Y YP(Y | T=1, X, I)\right) \\
&=\sum_{X,I}\sum_Y YP(Y | T=1, X, I)P(X,I) \\
&=\sum_{X,I}\sum_Y YP(Y | T=1, X, I)P(X,I) \frac{P(T=1|X,I)}{P(T=1|X,I)} \\
&=\sum_{X,I}\sum_Y  \frac{YP(Y, T=1, X, I)}{P(T=1|X,I)} \\
&= \sum_{X,I} \frac{\mathbb{E}(\mathbb{I}_{T=1, X, I}Y)}{P(T=1|X,I)} \\
&= \mathbb{E} \left(\frac{\sum_{X,I}\mathbb{I}_{T=1, X, I}Y}{P(T=1|X,I)}\right) \\
&= \mathbb{E} \left(\frac{\mathbb{I}_{T=1}Y}{P(T=1|X,I)} \right)\\
\end{aligned}
$$

\section{Additional Simulated Data Example}
\label{section: more sims}
In this section, we show a simulated example where the interference signature depends on $T_i$. We generate 10,000 datasets, each with 110,000 units, according to the data generating process outlined below. This is a modified version of the example in Section~\ref{subsec: sim data}. Note that the interference signature for unit $i$ depends on that unit's treatment status. In particular, units who are treated experience lower levels of interference than if they were untreated. In Figure~\ref{fig:toy_model_bias2}, we show the bias when $\text{TACE}_{TY}$ is estimated for this dataset without adjustment, adjusting only for the confounders $X$, and adjusting for both $X$ and $I$.
\begin{gather*}
X_i\in_R\{0.1, 0.2, 0.3, 0.4\},\qquad
L_i\in_R\{1,\ldots,10000\},\qquad
U_{L_i}\in_R\{0.1, 0.2, 0.3, 0.4\},\\
T_i\sim \mathcal{B}\left(X_i+U_{L_i}\right),\quad
\widetilde{I}_i=\frac{\sum_{j\neq i, L_j=L_i}T_j}{\sum_{j\neq i,L_j=L_i} 1},\quad
I_i=\max\left(0,\,\,\widetilde{I}_i - 0.1\,T_i\right),\quad
W_i\sim\mathcal{B}\left(I_i\right),\\
Y_i\sim \mathcal{N}\left(4X_i + T_i + 4W_i, 1\right)
\end{gather*}

This example also highlights the crucial role of the overlap condition in Def~\ref{Def: local}. Note that in the data generating process above, there can be no treated units with interference signature $I_i>0.9$. Therefore, the overlap condition is violated where $I>0.9$ with $P\left(T=1\mid X, I>0.9\right)=0$, and so we must exclude these units when estimating $\text{TACE}_{TY}$. In this case, we still recover an unbiased estimate of $\text{TACE}_{TY}$ for the full population because the treatment effect is constant. However, as in all applications of causal inference, we must rely on domain knowledge to assess the relevance of causal quantities estimated on subsets of the population. 

\begin{figure}[h] 
\centering
   \includegraphics[scale=0.7]{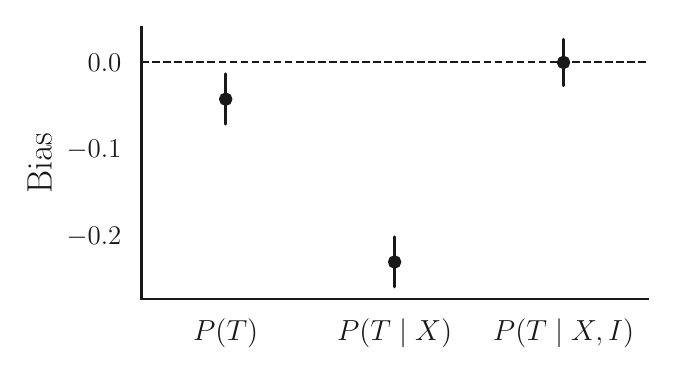}
      \caption{Bias in TACE$_{TY}$ estimates for the simulated data described in Appendix~\ref{section: more sims}. We show the bias when TACE$_{TY}$ is estimated without adjustment (IPW with $P(T)$), adjusting only for $X$ (IPW with $P(T\mid X)$), and adjusting for both $X$ and $I$ (IPW with $P(T\mid X, I)$).
      The error bars show the intervals between the 2.5\% and 97.5\% percentiles across 10,000 simulated datasets.
       }\label{fig:toy_model_bias2}
\end{figure}

\section{Variance Inflation Factor}
\label{section: real data appendix}
In this paper we focus on the identifiability of $\text{TACE}_{TY}$, and an investigation of procedures for estimation of confidence intervals is beyond the scope of the current work. 
However, conventional estimates of confidence intervals are likely to be overly optimistic in the presence of interference, and so in Section~\ref{subsec: real data} we employ a straightforward variance inflation procedure outlined by \cite{savje2021average}. 
This procedure relies on knowledge of the interaction network, and aims to construct \emph{conservative} confidence intervals for the IPW estimator in the presence of interference.

As discussed in Section~\ref{subsec: real data}, we assume that cantons interfere with adjacent cantons. Let $A_{ij}$ be the $N\times N$ adjacency matrix indicating which of the $N$ cantons are adjacent. For convenience, we set the diagonal elements to $A_{ii}=1$. Following \cite{savje2021average}, we define the number of interference dependencies for canton $i$ as
\begin{equation*}
D_i=\sum_{j=1}^N D_{ij}\qquad\text{where}\qquad D_{ij}=
\begin{cases}
1 & \text{if } A_{ki}\,A_{kj}=1 \text{ for any } k\in\{1,...,N\} \\
0 & \text{otherwise}
\end{cases}
\end{equation*}
The interference dependence indicator $D_{ij}$ captures whether cantons $i$ and $j$ are impacted by a common treatment --- either directly interfering with each other, or jointly interfering with a third canton. 
Note that in the absence of interference we would have $D_i=1$ for all cantons.
These quantities measure the deviation of an experiment from the zero interference case, and \cite{savje2021average} discuss how summaries of these can serve as adequate variance inflation factors. In particular, they consider the unit average dependence $D_\text{avg}=\frac{1}{N}\sum_{i=1}^N D_i$, the unit maximum dependence $D_\text{max}=\max_i D_i$, and the spectral radius of $D_{ij}$, which corresponds to the maximum of the absolute value of its eigenvalues, and falls between these $D_\text{avg}\leq D_\text{sr} \leq D_\text{max}$.
\cite{savje2021average} argue that $D_\text{avg}$ is generally overly optimistic, while $D_\text{max}$ is generally overly conservative, and therefore favour $D_\text{sr}$ as a compromise between these. In Figure~\ref{fig:covid_tace} we show the 95\% bootstrap confidence intervals, with the estimated variance inflated as $\text{Var}_\text{sr} = D_\text{sr} \,\text{Var}_\text{bootstrap}$. 
In this specific application, the confidence intervals are very wide due to large variance inflation. The inflation factor is quite large because many units are either adjacent or share a common adjacent unit. The variance inflation would not be as extreme in similar applications with sparser adjacency graphs.

\end{document}